\documentclass[envcountsame]{llncs}
\newcommand{\en}{\sl e}
\newcommand{\DL}{\textsc{Delivery-Test-on-the-Line\;}}
\newcommand{\DB}{\textsc{Delivery2-on-the-Line}}
\newcommand{\AM}{\textsc{AM}}
\newcommand{\CL}{\textsc{Convergecast-on-the-Line}}
\newcommand{\BL}{\textsc{Broadcast-on-the-Line}}

\newcommand{\RDC[1]}{\overrightarrow{\mathcal B}(#1)}
\newcommand{\RDB[1]}{\overleftarrow{\mathcal B}(#1)}
\date{}
  \usepackage[margin={1.85in, 1.55in}]{geometry}
\usepackage{verbatim}
  \usepackage[algo2e, noend, noline, boxed]{algorithm2e}
  \SetKwComment{tcm}{\{}{\}}

  \usepackage{enumerate}
  \usepackage{verbatim}
  \usepackage{stfloats}
  \usepackage{graphicx}
  \usepackage{bold-extra}
  \usepackage{amsfonts,amsmath}
  \usepackage{tikz}
  \usetikzlibrary{decorations.pathreplacing,calc,snakes}

  \spnewtheorem{observation}[theorem]{Observation}{\bfseries}{\itshape}

\def\convergecast{convergecast\xspace}
\def\broadcast{broadcast\xspace}

  \title{Algorithms for 
 Communication Problems \\for Mobile Agents  Exchanging Energy }
  \author{
    Jerzy Czyzowicz\inst{2}\thanks{
      Supported by NSERC grant of Canada
    }
    \and
    Krzysztof Diks\inst{1}\thanks{
      Supported by the grant NCN2014/13/B/ST6/00770 of the Polish Science Center.
    }
    \and
    Jean Moussi\inst{2}
    \and
  Wojciech Rytter
    \inst{1}\thanks{
      Supported by the grant NCN2014/13/B/ST6/00770 of the Polish Science Center.
    }}
\institute{
    Faculty~of Mathematics, Informatics and Mechanics,\\
    University of Warsaw, Warsaw, Poland\
    \email{[diks,rytter]@mimuw.edu.pl}
    \email{[jurek,Jean.Moussi]@uqo.ca}
    \and
    D\'epartement d'informatique, Universit\'e du Qu\'ebec \\ en Outaouais, Gatineau,
    Qu\'ebec,
    Canada
  }
\newcommand{\myskip}{\vskip 0.4cm \noindent}
\begin{document}
  \maketitle
  \begin{abstract}
We consider  communication problems in the setting of mobile agents deployed in an
edge-weighted network. The assumption of the paper is that each agent has some energy that it  can transfer to any 
other agent when they meet (together with  the  information it holds).
 The paper deals with three communication problems: data delivery,
convergecast and broadcast.
These problems are posed for a centralized scheduler which has full knowledge of the instance.
   It is already known that, without energy exchange, all three  problems are NP-complete even if
the network is
   a line.  Surprisingly, if we allow the agents to exchange energy,  we show that
 all three problems are polynomially
   solvable  on trees and  have linear  time algorithms on the
line. On the other hand for general undirected and directed graphs we show that
   these problems, even if energy exchange is allowed, are still NP-complete.
   \end{abstract}
\section{Introduction}\label{sect:intro}

A set of $n$ agents is placed at nodes of an edge-weighted graph $G$. An edge weight represents its length, i.e., the distance between its endpoints along the edge. Each agent has an amount of energy (possibly distinct for different agents). Agents walk in a continuous way along the network edges using amount of energy proportional to the distance travelled. 

An agent may stop at any point of a network edge (i.e. at any distance from the edge endpoints, up to the edge weight). Each agent has memory in which it can store information. 

When two agents meet, one of them can transfer a portion of currently possessed energy to another one. Moreover, two meeting agents exchange their currently possessed information, so that after the meeting each of them keeps in its memory the union of pieces of information previously hold by each of them. 

We assume that each agent has sufficient memory  to store the entire information initially belonging to all agents.

\noindent Our algorithms work as centralized schedulers having full knowledge of the instance.
\myskip
\noindent
We consider three problems:
\begin{enumerate}
\item \textbf{{\em Data delivery} problem:} Given two nodes $s,t$ of $G$, is it possible to transfer the initial packet of information placed at node $s$ to node $t$?
\item \textbf{{\em Convergecast} problem:} Is it possible to transfer the initial information possessed by each agent to a fixed agent ? (See Fig.1)
\item \textbf{{\em Broadcast} problem:} Is it possible to transfer the initial information of some agent to all other agents ? (See Fig.1)
\end{enumerate}

  \begin{figure}[htb]
\centering
\includegraphics[width=7.cm]{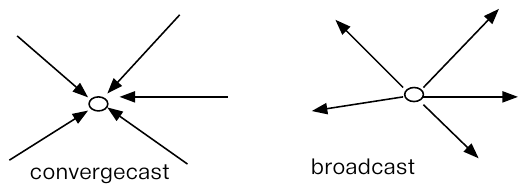}
 \label{general}
 \caption{Schematic view of convergecast and broadcast.}
\end{figure}
We will look for schedules of agent movements which will not only result in completing the desired task, but also attempt to maximize the unused energy. We call such schedules {\em optimal}.
We conservatively suppose that, whenever two agents meet, they automatically exchange the entire information they held. This information exchange procedure is never explicitly mentioned in our algorithms, supposing, by default, that it always takes place when a meeting occurs.

\subsection{Our results}\label{s:related}
We show that
 all three communication problems are polynomially
   solvable  on trees and  have linear  time algorithms on the
line. On the other hand for general undirected and directed graphs we show that
   these problems, even if energy exchange is allowed, are still NP-complete.

\subsection{Related work}\label{s:related}

Recent development in the network industry triggered  the research interest in mobile agents computing.  Several applications involve physical mobile devices like robots, motor vehicles or various wireless gadgets. Mobile agents are sometimes interpreted as
software agents, i.e., programs migrating from host to host in a network, performing some specific tasks.
Examples of agents also include living beings: humans (e.g. soldiers or disaster relief personnel) or animals. Most studied problems for mobile agents involve some sort of environment search or exploration (cf. \cite{AH,BS,DKS,FGKP}). In the case of a team of collaborating mobile agents, the challenge is to balance the workload among the agents in order to minimize the exploration time. However this task is often hard (cf. \cite{FHK}), even in the case of two agents in a tree, \cite{AB}.

The task of \convergecast is important when agents possess partial information about the network (e.g. when individual agents hold measurements
performed by sensors located at their positions)
and the aggregate information is needed to make some global decision based on all measurements. The {\convergecast} problem is often considered as a version of the data aggregation question (e.g. \cite{KEW,RV}) and it has been investigated mainly in the context of  wireless and sensor networks, where the energy consumption is an important issue (cf. \cite{AGS,KK}). 

The task of \broadcast is useful, e.g., when a designated leader needs to share its information with collaborating agents in order to perform together some future tasks. 

The \broadcast problem for stationary processors has been extensively studied both in the case of the message passing model, (e.g. \cite{AGP}), and for the wireless model, (see \cite{BGI}).


The power awareness question has been studied in different contexts. Energy management of (not necessarily mobile) computational devices has been studied in  \cite{Albers}. To reduce energy consumption of computer systems the methods proposed include power-down strategies (see \cite{Albers,AIS,ISG}) or speed scaling (cf. \cite{YDS}). Most of  research on energy efficiency considers optimization of overall power used. When the power assignments are made by the individual system components, similar to our setting, the optimization problem involved has a flavor of load balancing (cf. \cite{Azar}).


The \broadcast problem for stationary processors has been extensively studied both in the case of the message passing model, (e.g. \cite{AGP}), and for the wireless model, (see \cite{BGI}).

The problem of communication by energy-constrained mobile agents has been investigated in \cite{Anaya}. The agents of \cite{Anaya} all have the same initial energy and they perform efficient convergecast and broadcast in line networks. However the same problem for tree networks is proven to be strongly $NP$-complete in \cite{Anaya}. 

The closely related problem of data delivery, when the information has to be transmitted between two given network nodes by a set of energy constrained agents has been studied in \cite{NP}. This problem is proven to be $NP$-complete in \cite{NP} already for line networks, if the initial energy values may be distinct for different agents.  However, in the setting studied in \cite{Anaya,NP}, the agents do not exchange energy. In the present paper we show that the situation is quite different if the agents are allowed to transfer energy between one another.

\section{The line environment}

In this section we suppose that we are given a collection of agents $\{0,1, 2, \ldots ,n-1\}$  on the line.  Each agent $i$ is initially placed at position $a_i$ on the line and has initial energy $\en(i)$. 
We investgate delivery, convergecast and broadcast problems separately. The are solved using auxiliary
table.s

\subsection{Data delivery on the line}\label{sect:data-line}
We start with the delivery problem from point $s$ to $t$. Assume w.l.o.g. that $a_i < a_j$ for $i<j$ and $s<t$ .

The problem can be immediately reduced to the situation
$s=a_1,\; t=a_n.$
 Otherwise  the first agent is going to $s$ from left to right, swallawing energy of
encountered agents, symetrically the righmost agent is going right to left until reaching $t$. Then we can 
reduce $n$ and renumber agents setting $a_1=s,\; a_n=t$.

\noindent 
\vskip 0.1cm
\noindent Our first algorithm is only a decision version. Its main purpose is to show how certain useful table can be computed, 
all other algorithms are based on computing similar type of tables. 
\myskip
Consider the partial delivery problem ${\mathcal{D}}_i$ which is the original problem with agents larger than $i$ removed,
together with their energy, and the goal is to deliver the packet from the 1-st agent to the $i$-th agent.
\myskip
We say that the problem ${\mathcal{D}}_i$ is {\em solvable} iff such a delivery is possible.
\myskip
We define the following table $\overrightarrow{\Delta}$:
\begin{itemize}
\item If ${\mathcal{D}}_i$ is not solvable then $\overrightarrow{\Delta}(i)\,=\, -\delta$, 
where $\delta$ is the {\em minimal} energy which needs to be added
to $e_i$ (to energy of $i$-th agent) to make ${\mathcal{D}}_i$  solvable.
\item
If ${\mathcal{D}}_i$ is solvable then $\overrightarrow{\Delta}(i)$ is the {\em maximal} energy which can remain in point $a_i$ after delivering
the packet from $a_1$ to $a_i$. Possibly $\overrightarrow{\Delta}(i)>e(i)$ since during delivery the partial energy of some
other agents can be moved to point $a_i$.
\end{itemize}
\noindent 
%
\vskip 0.1cm \noindent Assume in the algorithm that points $s$ and $t$ are the starting points $s=a_1$ and $t=a_n$.
In our algorithm the statements of the form $x+=y$ are equiavalnet to $x:=x+y$.
  \vskip 0.2cm  \begin{small}
    \begin{center}
    \fbox{\vspace*{0.2cm}
    \begin{minipage}{10cm}
    \vspace*{0.3cm} \noindent
     \hspace*{0.2cm} {\bf ALGORITHM} \DL;
\vskip .2cm
\hspace*{0.5cm} \{\,Decision version and computation of table $\overrightarrow{\Delta}$\,\}
\vskip 0.2cm
1.  \hspace*{0.3 cm} {\bf for each } $i\in [0..n-1]$ {\bf do} $\overrightarrow{\Delta}(i):=\en(i);\, $\vskip .3cm
2. \hspace*{0.3 cm}  {\bf for} $i=1$ to $n$  {\bf do}\vskip .3cm
3. \hspace*{0.8cm} $d:= a_i-a_{i-1};$  \vskip 0.3cm
4. \hspace*{0.8cm} {\bf  if} $\overrightarrow{\Delta}(i-1) \geq d$  {\bf  then} $\overrightarrow{\Delta}(i)\;\mbox{+=} \;\overrightarrow{\Delta}(i-1)-d$\vskip 0.3cm
5. \hspace*{0.8cm}   {\bf  else if} $\overrightarrow{\Delta}(i-1) \geq 0$ {\bf  then} 
$\overrightarrow{\Delta}(i)\;\mbox{+=} \;2 (\overrightarrow{\Delta}(i-1)-d)$ \vskip 0.3cm
6. \hspace*{0.8cm}   {\bf  else} $ \overrightarrow{\Delta}(i)\;\mbox{+=}\;\overrightarrow{\Delta}(i-1)-2d$; 
\vskip 0.3cm
8. 
 \hspace*{0.2cm} Delivery from $a_1$ to $a_n$ is possible iff $\overrightarrow{\Delta}(n)\geq 0$

 \vspace*{0.2cm}
  \end{minipage}
  }
  \end{center}
  \end{small}
\myskip
\begin{example}\label{example1}
Assume $[a_1,a_2,\ldots a_5]\;=\; [0,\, 10,\, 20,\, 30,\, 40,\, 50]$, and \vskip 0.1cm
$[e(1),e(2),\ldots e(5)]\;=\; [0,\, 24,\, 10,\, 40,\, 0].$
Then 
$\overrightarrow{\Delta}\;=\;[0,\,4,\, -2,\, 18,\,8].$
\end{example}
%
\vskip 0.1cm
  \noindent {\bf Remark.} The values of $\overrightarrow{\Delta}(i)$ are not needed to
solve the decision-only version. However they will be useful
  in creating the delivery schedule and also in the {\em convergecast} problem.
\vskip 0.1cm
  \noindent
  \begin{lemma}
  \label{lm:ass}
The algorithm \DL correctly computes the table $\overrightarrow{\Delta}$ (thus it solves the decision
version of the delivery problem) in linear time.
  \end{lemma}
 \begin{proof} 
Asume the algorithm computed correctly $\overrightarrow{\Delta}(i-1]$. There are three
cases:
\begin{description}
\item {{\bf Case 1} (statement 4).} 

The instance ${\mathcal{D}}_{i-1}$ is solvable and after moving 
the packet from $a_1$ to $a_{i-1}$ the maximal remaining energy is $\overrightarrow{\Delta}(i-1]$.
In this case this energy is sufficient to move the packet from $a_{i-1}$ to $a_i$, together
with remaining energy. We spent $d$ energy to travers the distance $d$ in one direction.
Then we get in total ${\mathcal{D}}_{i-1}+{\mathcal{D}}_{i}-d$ energy in $a_i$.
\vskip 0.2cm
\item {{\bf Case 2} (statement 5).

} The instance ${\mathcal{D}}_{i-1}$ is still solvable but after moving
the packet from $a_1$ to $a_{i-1}$ the remaining energy is not sufficient to reach $a_i$ without
{\em help} from agents to the right of $a_{i-1}$. Then the $(i-1)$-st agent moves only
one-way by distance $\overrightarrow{\Delta}(i-1]$. The remaining distance  (to $a_i$) should be covered both-ways
from $a_i$. Hence $a_i$ looses $2(d-\overrightarrow{\Delta}(i-1])$ energy, which is expressed by statement 5.
\vskip 0.2cm
\item {{\bf Case 3} (statement 6).

} The instance ${\mathcal{D}}_{i-1}$ is not solvable,
now the node $a_{i-1}$ needs $|\overrightarrow{\Delta}(i-1]|$ additional energy which has to be delivered 
from $a_i$. Additionally the $i$-th agent needs $2d$ energy to cover the distance $d$ from $a_i$ to
$a_{i-1}$ in both directions. Consequently its energy is reduced by $2d+|\overrightarrow{\Delta}(i-1]|$.
This is reflected in the statement 6. 
\end{description}
\noindent The cases correspond to the statements in the algorithm, and
show their correctness. This completes the proof.
\end{proof}

\noindent 
Once the values of $\overrightarrow{\Delta}(i)$ are computed, the schedule describing the behavior of each agent is implicitly obvious, but we give it  above for reference. Note that the action of each agent $a$ is started once the process involving lower-numbered agents has been completed. We are not interested in this paper in finding the time to complete the schedule allowing agents to work in parallel.
\vskip 0.1cm
{\bf Observation.} \noindent In the scheduling algorithm if an agent has the packet (in particular if it is the 
first agent) then it is not going left, since it holds the packet already.
  \begin{small}
    \begin{center}
    \fbox{\vspace*{0.2cm}
    \begin{minipage}{11.6cm}
    \vspace*{0.2cm} \noindent
     \hspace*{0.2cm} {\bf  ALGORITHM} {\textsc{Delivery-Schedule-on-the-Line}} ;
\vskip .3cm
\hspace*{0.3cm} \{\, Delivering packet from $a_1$ to $a_n$\,\}
\vskip 0.3cm
\hspace*{0.3 cm} $pos:=a_1$;\vskip .2cm
 \hspace*{0.3 cm}  {\bf for} $i=1$ to $n$  {\bf do}\vskip .3cm
 \hspace*{0.7cm} {\bf  if} $\overrightarrow{\Delta}(a) \geq 0$ and $pos\le  a_i$  {\bf  then}\vskip .3cm
 \hspace*{1.2 cm} 1. The $i$-th agent walks left swallowing energy of  encountered \\
\hspace*{2 cm} 
agents until arriving
at the packet position. It picks up the packet.\vskip .3cm
 \hspace*{1.2 cm} 2. The $i$-th agent walks right swallowing energy of encountered \\
\hspace*{2 cm}agents  until  exhausting his energy or reaching $a_n$. \vskip .3cm
 \hspace*{1.2 cm}
3. Leave the packet at  th actual position $pos$ of the $i$-th agent. 
\vskip 0.3cm
\hspace*{0.3 cm} 
 Delivery is successful iff $pos=a_n$; \vspace*{0.2cm}
  \end{minipage}
  }
  \end{center}
  \end{small}
  \vskip 0.2cm \noindent
\vskip 0.2cm
     \begin{figure}[htb]
\mbox{ \ }
\centering
\includegraphics[width=12.0cm]{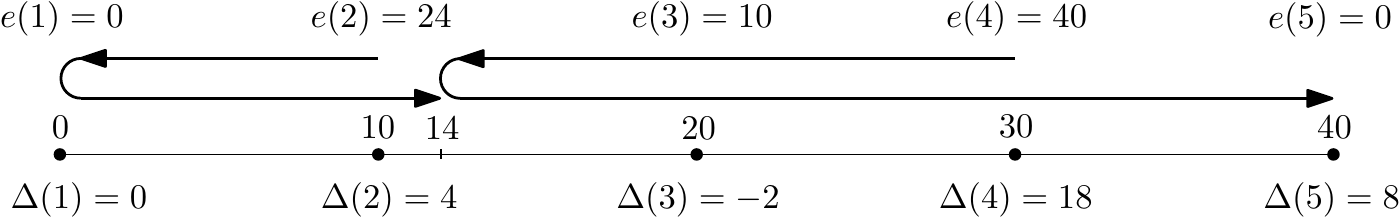}
\caption{Schedule of agent movements for $a_i$'s and energies given in Example 1.
} \label{przyklad}
 \vspace{-0.4cm}
\end{figure}
\vskip 0.3cm
\noindent We conclude with the following theorem.
\begin{theorem}
In linear time we can decide if the information of any agent can be delivered to any other agent
and if it is possible find the centrelized scheduling algorithm
which perorms such a delivery.
\end{theorem}

In the next section we show that the solution may be extended to the convergecast problem.


\noindent

\subsection{The convergecast on the line}

The convergecast consists in   communication in which the union of initial information of all agents is in the hands of the same agent. Sometimes the problem consists in verifying whether a particular agent is a convergecast agent and other times it has to be determined whether any such agent exists. For energy exchanging agents, if convergecast is possible any agent may be its target, as agents may swap freely when meeting.

We say that a schedule results in convergecast iff there exists an agent $c$, for $1 \leq c \leq n$ (called a {\em convergecast agent}), and a schedule permitting to collect by agent $c$ the initial information of any other agent.

  \begin{small}
    \begin{center}
    \fbox{\vspace*{0.2cm}
    \begin{minipage}{11.2cm}
    \vspace*{0.3cm} \noindent
     \hspace*{0.2cm} {\bf ALGORITHM} \CL;
     \vskip .2cm
1.  \hspace*{0.3 cm}  Compute the values of $\overrightarrow{\Delta}(i)$ and 
$\overleftarrow{\Delta}(i)$ representing the energy \\
 \hspace*{0.65 cm}   potentials at $a_i$ for deliveries from $a_1$ to $a_n$ and $a_n$ to $a_1$, respectively\vskip .3cm
2. \hspace*{0.3 cm}  {\bf for} $i=1$ to $n$  {\bf do}\vskip .3cm
3. \hspace*{0.8cm} $d:=a_{i+1}-a_i;$ \vskip 0.3cm
4. \hspace*{0.8cm} {\bf  if} $(\overrightarrow{\Delta}(i) \geq 0) \wedge
     ( \RDC[i+1] \geq 0)$  {\bf  then}
     $E:=\overrightarrow{\Delta}(i) + \overleftarrow{\Delta}(i) - d;$  \vskip 0.3cm
5. \hspace*{1.32cm} $(\overrightarrow{\Delta}(i) \geq 0) \wedge
    ( \overleftarrow{\Delta}(i) < 0)$  {\bf  then}
    $E:=\overrightarrow{\Delta}(i) + \overleftarrow{\Delta}(i)/2  - d;$  \vskip 0.3cm
6. \hspace*{1.32cm} $(\overrightarrow{\Delta}(i) < 0)  \wedge
    ( \RDC[i+1] \geq 0)$ {\bf  then}
    $E:=\overrightarrow{\Delta}(i)/2 + \overleftarrow{\Delta}(i) - d;$  \vskip 0.3cm
7. \hspace*{0.8cm}    {\bf  if} $E \geq 0$ {\bf  then} {\bf return} Convergcast possible;  \vskip 0.3cm
8. \hspace*{0.3 cm}  {\bf return} Convergcast not possible;\vskip 0.3cm

 \vspace*{0.3cm}
  \end{minipage}
  }
  \end{center}
  \end{small}
  \vskip 0.3cm

We have the following theorem.
\begin{theorem}
\label{CL}
Algorithm \CL~in $O(n)$ time solves the convergecast problem.
 \vspace{-0.2cm}

\end{theorem}
 \vspace{-0.2cm}

\subsection{The broadcast on the line}
In the broadcast communication one agent has to transfer its original information to all other agents of the collection. In this section we present an algorithm determining which agents are able to broadcast. Contrary to the convergecast problem, for the energy transferring agents, only a selected subset of them may be able to perform broadcast. However the approach used for convergecast may be transformed the way it is also useful for broadcast communication. We design the schedule, which not only performs broadcast, whenever possible, but also tries to use as little energy of the broadcasting agent as possible. We call such schedule {\em optimal}.

Our algorithm will compute for all agents the values of $\RDC[i]$ and $\RDB[i]$. The value of $\RDC[i]$ equals the potential of energy at point $a_i$ for the delivery of initial information of agent $i$ to agent $1$ using only agents $1, 2, \dots, i$. More exactly, if $\RDC[i] < 0$, then $-\Delta(i) $ equals the minimal amount energy which must be added to $e(i)$ so that the delivery from $a_i$ to $a_1$ is possible.
may be suppressed from $\en(i)$ so that the delivery from $a_i$ to $a_1$ be possible.

The values of $\RDB[i]$ are defined symmetrically.

The following algorithm may be viewed as a different version solving the delivery problem. However the main reason of its presentation is to compute $\RDC[i]$ values, which are used for broadcast schedule.

    \begin{small}
    \begin{center}
    \fbox{\vspace*{0.2cm}
    \begin{minipage}{10.8cm}
    \vspace*{0.3cm} \noindent
     \hspace*{0.2cm} {\bf ALGORITHM} \DB;
     \vskip .2cm

1.  \hspace*{0.3 cm}  $e:=0; s:=a_1;$\vskip .3cm
2. \hspace*{0.3 cm}  {\bf for} $i=1$ to $n+1$  {\bf do}\vskip .3cm
3. \hspace*{0.8cm} {\bf  if} $a_i > \en(i)+e + s$  {\bf  then} $e := e+\en(i)$; $\RDC[i] :=s+e- a_i$;\vskip 0.3cm
4. \hspace*{0.8cm}   {\bf  else}\\
\hspace*{2cm}  $s\; +=\; {(e+\en(i)+ \max(a_i-s,0))}/{2}$;\\
\hspace*{1.8cm} \; $e:=0;\; \; \RDC[i]:=2(s-a_i)$;\vskip 0.3cm 
5. \hspace*{0.2cm} {\bf return } {\em true} iff $(s \geq a_n)$ ; \vskip .3cm

  \end{minipage}
  }
  \end{center}
  \end{small}
  \vskip 0.2cm

   \begin{lemma}
  \label{lm:BP}
\mbox{ \ }\\
{\bf (a)} If $\RDC[i] \geq 0$ in algorithm \DB~then agents $1,2,\ldots, i$ can deliver the packet from agent $i$ at initial position $a_i$ to all agents $1,2,\ldots, i-1$. Moreover, $\RDC[i]$ equals the maximal amount of energy which may be suppressed from $\en(i)$ (and $\en(i-1), \en(i-2), \ldots , \en(1)$, if possible, in this order) so that this delivery is still possible.
\vskip 0.2cm \noindent
{\bf (b)} The tables $\RDC[i],\,  \RDB[i]$ can be computed in linear time.
  \end{lemma}
  \begin{proof}
  We prove by induction that the following assertion is true at the completion of $i$-th iteration of the "for" loop of the algorithm:\
 { \sl
 Let $k$ be the largest integer $k \leq i$, such that agents $1, 2, 3 \dots, k$ can deliver the packet from point $s$, such that $s \geq a_k$, to point $a_1$ and $s$ is the largest possible point having this property. Then
  \vspace*{0cm}
  \begin{enumerate}
\item either $k=i$, $s \geq a_i$, $e=0$ and $\RDC[i] = 2(s-a_i) \geq 0$,
\item or $k<i$, $s < a_i$ and $\RDC[i]=s+e-a_i <0$, where $e=\sum_{i=k+1}^{i-1} \en(i)$
\end{enumerate}
}
The assertion is clearly true for $i=1$ as $a_1=s, e=0$ and $\RDC[1]=0$.
Suppose the assertion was true after iteration $i-1$ and all previous iterations.

Consider first the case 1, i.e. that $k=i-1$ agents could deliver to $a_1$ the packet from point $s \geq a_{i-1}$ and at the completion of the $(i-1)$-th iteration $e=0$ and $DC[i-1]s-a_{i-1})$. Suppose first that at the $i$-th iteration the condition at line 3 is true, i.e.  $a_{i} > \en(i)+e + s$. That means that agent $1$  cannot deliver the packet from point $a_{i}$ to $s$ and $\RDC[i] =s+e-a_{i} <0$, so the case 2 of the assertion is true. Suppose now that $a_{i} \leq \en(i)+e + s$ at line 3 of the $i$-th iteration. The agents $1,2, \ldots, k$ could deliver the packet from $s \geq k$ to $a_1$. Observe that $e$ contains the accumulated energy of agents ${k+1}, {k+2}, \ldots, {i-1}$, i.e.
$$
e=\sum_{i=k+1}^{i-1} \en(i)
$$
Therefore the condition  $a_{i} \leq \en(i)+e + s$ implies that agents $k+1, \ldots,i$  can deliver the packet from point $a_{i}$ to $s$. Hence agents $1,2, \ldots, i$ can deliver from  $a_{i}$ to $a_1$. Since $e=0$ and $\RDC[i] = 2(s-a_{i}) \geq 0$, the case 1 of the assertion becomes true. Observe that the new point $s$, computed at line 4 is at distance equal twice the excess of energy obtained at $a_i$ so it is computed correctly.

Take now the case that $i-1$ agents could not deliver to $a_1$ the packet from point $a_{i-1}$ and at the completion of the $(i-1)$-th iteration we have $s < a_i$ and $\RDC[i] =s+e-a_i <0$. Suppose first that at the $i$-th iteration the condition at line 3 is true, i.e.  $a_{i} > \en(i)+e + s$. Then, it means that the packet cannot be delivered from $a_i$ to $s$ using agents ${k+1}, {k+2}, \ldots, {i-1}$ and in line 3 we have $e=\sum_{i=k+1}^{i} \en(i)$, so the case 2 of the assertion is true. Suppose now that $a_{i} \leq \en(i)+e + s$ at line 3 of the $i$-th iteration. Similarly to the case above, agents $1,2, \ldots, k$ could deliver the packet from $s \geq a_k$ to $a_1$ and agents  $k+1, \ldots,i$  can deliver the packet from point $a_{i}$ to $s$ and the case 1 of the assertion becomes true. This completes the proof.

  \end{proof}

  \begin{small}
    \begin{center}
    \fbox{\vspace*{0.2cm}
    \begin{minipage}{10.8cm}
    \vspace*{0.3cm} \noindent
     \hspace*{0.2cm} {\bf ALGORITHM} \BL;
     \vskip .2cm

     1.  \hspace*{0.3 cm}  $BR:= \emptyset;$\vskip 0.3cm
     2. \hspace*{0.3 cm} Compute the values of $\RDC[i]$ and $\RDB[i]$ representing the energy \\
 \hspace*{0.65 cm}   potential for deliveries from $a_1$ to $a_n$ and $a_n$ to $a_1$, respectively\vskip .3cm
3. \hspace*{0.3 cm}  {\bf for} $i=1$ to $n$  {\bf do}\vskip .3cm
4. \hspace*{0.8cm} {\bf  if} $\RDC[i]+\RDB[i+1] - 2(a_{i+1}-a_i) \geq 0$  {\bf  then}   \vskip 0.3cm
5. \hspace*{1.4cm} $BR :=BR \cup \{i,i+1\}$;\vskip 0.3cm
6. \hspace*{0.3 cm}  Report $BR$ as set of broadcast agents;\vskip 0.3cm

 \vspace*{0.2cm}
  \end{minipage}
  }
  \end{center}
  \end{small}
  \vskip 0.2cm

\begin{theorem}
\label{thm:broadcast}
 \vspace{-0.1cm}
Algorithm \BL~identifies  all agents of the collection, which are able to broadcast.
\end{theorem}
  \begin{proof}
  Note that $x=a_i + \RDC[i]/2$ is the rightmost point from which the agents $1,2, \dots, i$ can pick up the packet and broadcast in the left direction so that the packet reaches agent $1$. Also $y=a_{i+1} -\RDB[i+1]/2$ is the leftmost point from which the agents $i+1, \dots, n$ can pick up the packet and broadcast in the right direction so that the packet reaches agent $n$. Agents $i$ and $i+1$ can communicate during this broadcast off $x \geq y$, i.e. when  $\RDC[i]+\RDB[i+1] \geq 2(a_{i+1}-a_i)$ which is verified at line 4 of the algorithm. All such pairs of agents are included in the broadcasting set $BR$.
  \end{proof}

Observe that, the positive amounts of $E$ in \CL~and $\RDC[i]+\RDB[i]-\en(i) \geq 0$ in \BL~equal the maximal amount of energy which may be suppressed from agent $i$ and the corresponding communication is still possible.

  \section{The tree environment}\label{sect:tree}

We suppose that the agents are placed at the nodes of the undirected  tree. Observe that for each problem: the data delivery, convergecast and broadcast the tree may be truncated so that each leaf contains an initial position of an agent. Indeed, in neither problem it makes sense to visit subtrees containing no energy source.

We developed basic ideas and tables for the line, now they can be extended
to trees. The tables $\Delta$ for lines were computed locally, looking only
at neighboring agents. Simiarly for trees, the value of the corresponding tables
for a  node in a tree is computed looking at neighbors of this node.

\subsection{Data delivery in the tree}

The delivery problem for a tree is easily reducible to the case of a line.

\begin{theorem}
We can solve delivery problem and construct delivery-scenario on the tree in linear time.
\end{theorem}
\begin{proof}
We find the path $\pi$ connecting $s$ with $t$ in a given undirected tree $T$.
Suppose we remove edges of the path. The tree splits into several subtrees {\em anchored}
at vertices of $\pi$. For each such subtree we direct all edges towards the root, which is
a vertex of $\pi$. The leaves in these trees are sending their energies towards their
roots accumulating energies of intermediate nodes.
This way we reduce the delivery on the tree to the delivery on the line $\pi$.
This completes the proof.
\end{proof}

\subsection{Convergecast in the tree}

We can reduce the convergecast and broadcast problems on general
trees to trees with degree at most 3.
If a vertex $v$ has neighbors $v_1,..,v_k$ we can
change it locally to small binary internal subtree with leaves $v_1,..,v_k$ by
adding several edges with zero cost.
Hence we assume from now that in our undirected tree each vertex has at most 3 neighbors.

Though the input tree is undirected, we can consider direction of edges.
For each undirected edge $(u,v)$ we consider two directed edges $u\rightarrow v,\; v\rightarrow u.$

\noindent Define the subtree $T(u,v)$ as the connected component containing  $u$
and  resulting from $T$ by removing the edge $(u,v)$, see Figure~\ref{fig1}.

For each directed edge $u\rightarrow v$ of the tree we define $\mathrm{\Delta}(u,v)$
as the cost of moving
all packets from the subtree $T(u,v)$ to its root $u$ without interacting with any node
outside $T(u,v)$. The table $\mathrm{\Delta}$ is a generalization of the table $\Delta$ to trees.

 \begin{figure}[h]
\mbox{ \ }
\centering
\vspace*{-0.2cm}
\includegraphics[width=7.5cm]{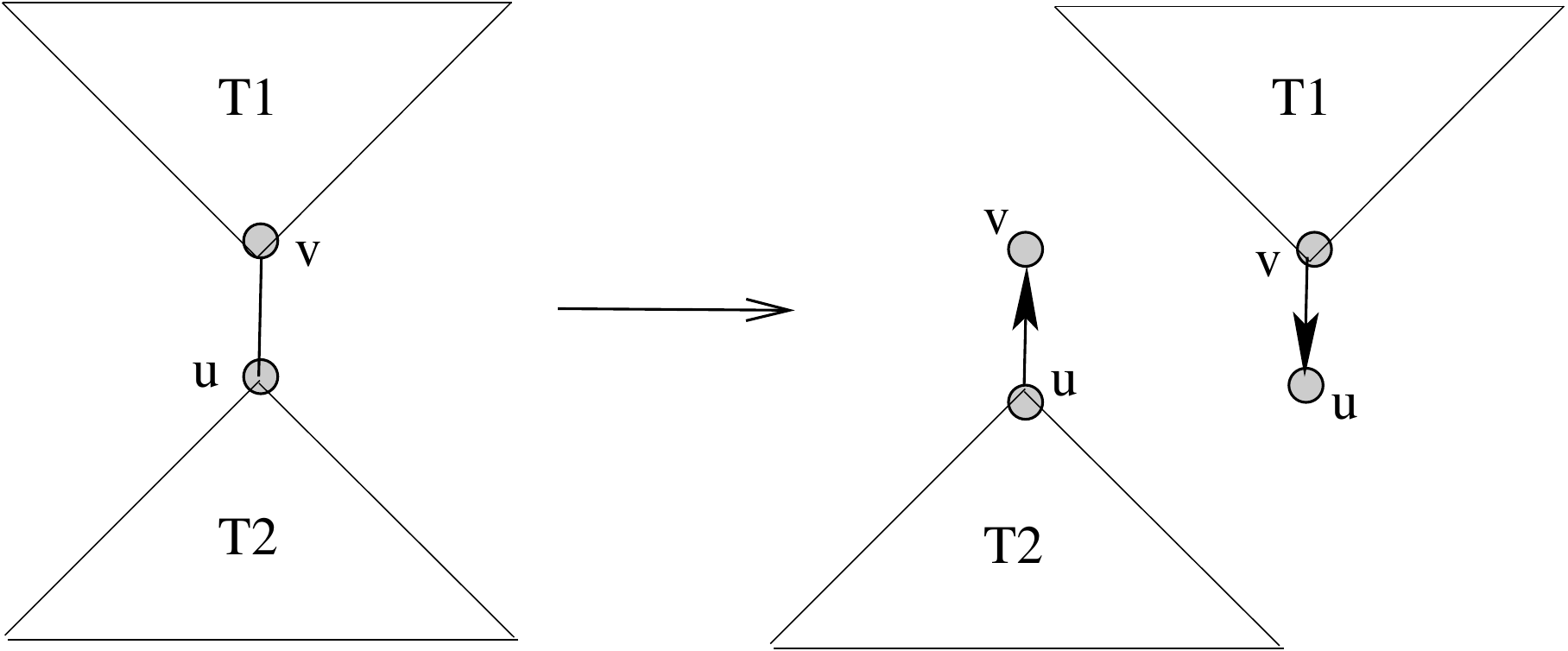}
\caption{Testing if there is
a {\em convergecast point} on the undirected edge $(u,v)$ is
reduced to computation of the costs $\mathrm{\Delta}(u,v)$ and $\mathrm{\Delta}(v,u)$ of moving all packets in the trees
$T2=T(u,v)$ and $T1=(v,u)$ .
} \label{fig1}
\vspace*{-0.4cm}
\end{figure}

We can use similar algorithm as computing locally values $\overrightarrow{\Delta}(v)$ on the line
in the previous section, consequently we have the following fact.

\begin{lemma}\label{xx}\mbox{ \ }\\
{\bf (a)}
Assume $u,u_1,u_2$ are neighbors of the vertex $v$.
Then, knowing $\mathrm{\Delta}(u_1,v),\, \mathrm{\Delta}(u_2,v)$ we can compute $\mathrm{\Delta}(v,u)$
in constant time.\\
{\bf (b)} If $\mathrm{\Delta}(v,w),\;\mathrm{\Delta}(w,v)$ are known then we can compute in constant time
all convergecast points on the edge $(v,w)$ as its subinterval.
\end{lemma}
\vskip 0.2cm
Let us root the tree $T$ at some vertex $root$ obtaining a rooted version $\overline{T}$ of $T$.

\vspace*{-0.2cm}
\begin{observation}
For each node $v\in \overline{T},\; v\ne root$, there are directed edges
$v\rightarrow father(v),$ (outgoing) and $father(v)\rightarrow v$ (ingoing edge).
Assume we are to mark all directed edges, but we have to obey the rule:
if an edge $v\rightarrow w$ is marked then all edges ingoing into $v$ from nodes other than $w$
should be marked.
We can obey the rule and mark all edges by first traversing the tree in postorder marking
for each visted node $v$ the edge $v\rightarrow father(v)$.
Then we can traverse in preorder and for each visited node $v$ mark the edge $father(v)\rightarrow v$.
\end{observation}

We can replace the operation of {\em marking} an edge $e$ by the computation of
table $\mathrm{\Delta}(e)$. Using the observation above our algorithm computing convergecast points on $T$ can
 be written as the following pseudocode:

\begin{small}
    \begin{center}
    \fbox{\vspace*{0.1cm}
    \begin{minipage}{10.2cm}
    \vspace*{0.2cm} \noindent
     \hspace*{0.2cm} {\bf ALGORITHM} Convergecast on the tree $T$;
     \vskip .1cm
  \hspace*{0.3 cm}
  $root$:= any node of $T$;
  $\overline{T}$:= the directed version of $T$ rooted  at $root$;
   \vskip .1cm
 \hspace*{0.3 cm} {\bf for each} node $v\ne root$
of $\overline{T}$ in {\em postorder}  {\bf do}\vskip .1cm
\hspace*{0.8cm} Compute $\mathrm{\Delta}(v,father(v))$ \vskip 0.1cm
\hspace*{0.3 cm}  {\bf for each} node $v\ne root$
of $\overline{T}$ in {\em preorder}  {\bf do}\vskip .1cm
\hspace*{0.8cm} Compute $\mathrm{\Delta}(father(v),v)$  
\vskip .1cm
\hspace*{0.3cm} {\bf for each} undirected edge $(v,w)$
of ${T}$   {\bf do}\vskip .1cm
 \hspace*{0.8cm} compute in constant time, knowing $\mathrm{\Delta}(v,w),\;\mathrm{\Delta}(w,v)$
\vskip .1cm
\hspace*{0.8cm} all convergcast points on the edge $(v,w)$ as its subinterval;
 \vskip .1cm
 \hspace*{0.2cm} {\bf return } the set of convergcast points as sets of subintervals of edges;
 \vspace*{0.1cm}
  \end{minipage}
  }
  \end{center}
  \end{small}
  \vskip 0.2cm
The next theorem follows now from Lemma~\ref{xx}.
\begin{theorem}
The convergecast problem for undirected trees can be solved in linear time.
\end{theorem}

\vspace*{-0.55cm}

\subsection{Broadcasting in the tree}

We use terminology from the previous section and ideas from broadcasting on line.
Define ${\mathcal B}(u,v)[i,j]$ as the potential of energy at the vertex $u$
for the delivery of initial information of $i$ agents placed initially
at $u$ to all agents in subtree $T(u,v)$,
using only agents in $T(u,v)$, assuming that at the end we have $j$ agents at $u$.
It is similar to the definition of ${\mathcal B}$ for broadcasting on  the line.
Hence each ${\mathcal B}(u,v)$ is a $n\times n$ table.
The use of quadratic tables results in nonlinear algorithm.

\begin{observation}
The main innovation here is introducing many agents at the same node.
For example suppose we have a tree with single branch of $k$ nodes, each with single agent, from $r$ to $u$. The edges of this branch have zero cost. However $v$ has $k$ outgoing edges with large cost.
Then optimal broadcasting of such a tree from $r$ first gathers all $k$ agents at node $v$, then
each agent travels a different edge from $v$.
\end{observation}

\begin{theorem}
The broadcasting on a tree can be done in polynomial time.
\end{theorem}
{\bf Sketch of the proof}\
 We can show the following claim, analogous to Lemma~\ref{xx}.

\begin{claim}\label{xxx}\mbox{ \ }\\
{\bf (a)}
Assume $u,u_1,u_2$ are neighbors of vertex $v$.
\
Then we can compute ${\mathcal B}(v,u)$
in polynomial time, knowing tables ${\mathcal B}(u_1,v),\, {\mathcal B}(u_2,v)$.\\
{\bf (b)} If tables ${\mathcal B}(u,w)$ are known for each neighbor of $u$
then we can check in polynomial time if $u$ is a valid source of broadcast.
\end{claim}

The tables ${\mathcal B}$ for each edge are computed locally in a similar way as values of $\mathrm{\Delta}$
for edges, traversing first the rooted version of the
tree in postorder, and then in preorder.

Each edge can be processed in polynomial time,
due to the claim, hence the whole algorithm works
in polynomial time. We omit the details.

\section{NP-Completeness for digraphs and graphs }
We use the following NP-complete problem:
\myskip
{\bf Integer Set Partition:} For a given set $X$ with integer weights
check if $X$ can be partitioned into two disjoint subsets with
equal total sums of weights.

\begin{theorem}
The delivery problem is NP-complete for general directed graphs.
\end{theorem}
\begin{proof}
Consider the graph of the form shown in Figure~\ref{graph2}. The set of nodes consists
of the set of middle  nodes (the set $X$) and three additional nodes $s,t,a$. Initially there are agents only in middle nodes.

\noindent Denote by $E$ the total energy of all agents. Hence $E\,=\, \alpha+\beta.$
\myskip
Energy of an element $x\in X$ is set to be $e(x)=w(x)$. There is no energy in other points.
The length of en edge $x\rightarrow s$ equals $w(x)/3$.
The length of the edge $s \rightarrow a$ is $E/3$ and the
length of the edge $a\rightarrow t$ is $E/2$.
The lengths of other edges are zero.

 \begin{figure}[h]
\begin{center}
\mbox{ \ }
\centering
\vspace*{-0.25cm}
\includegraphics[width=8cm]{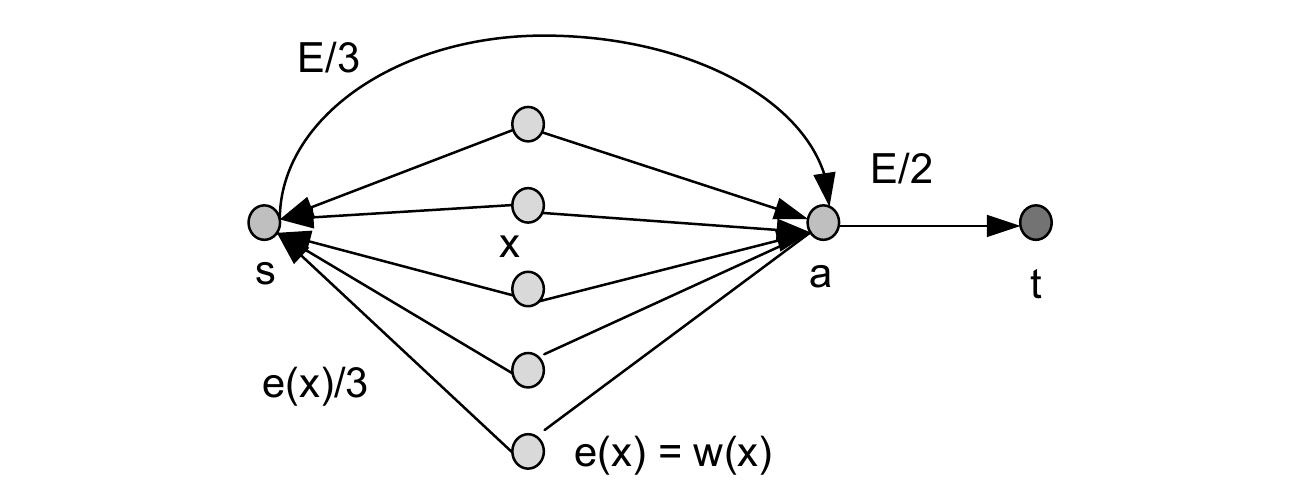}
\caption{The energy of all three non-middle nodes is zero. The weight
$w(x)$ of each middle node $x$ equals agent's energy $e(x)$.
Delivery from $s$ to $t$ is possible only iff
the set of weights can be partitoned into two sets of the same sum.}\label{graph2}
\end{center}
\vspace*{-0.65cm}
\end{figure}
%
\noindent  Assume $X_1,\, X_2$ are the sets of agents which move to $s$ and $a$, respectively.
Let $$\alpha\;=\; \sum_{x\in X1}\, w(x),\ \beta\;=\; \sum_{x\in X2}\, w(x)$$
Total energy coming to $s$ is $\frac{2}{3}\alpha$ and total energy coming to $a$
is
$$
\frac{2}{3}\alpha - (\alpha+\beta)/3 + \beta\;=\; \alpha/3 +\frac{2}{3}\beta
$$
Delivery from $s$ to $t$ is possible if and only if
$$\frac{2}{3}\alpha\; \ge\;  E/3\;=\; (\alpha+\beta)/3 \ \ \mbox{  and  } \ \ \alpha/3 +\frac{2}{3}\beta\;
\ge\;  E/2\;=\;   (\alpha+\beta)2$$
It is equivalent to:
$$\alpha\ge \beta \ \ \mbox{  and  } \ \  \beta\ge \alpha.$$
Consequently $\alpha=\beta$. In other words delivery from $s$ to $t$ is possible
iff the integer partition problem is solvable.
$NP$-completeness of the delivery problem follows from $NP$-completeness of
the integer partition problem.
\end{proof}
\begin{theorem}
The delivery, convergecast and broadcast problems are
 NP-complete for general undirected graphs.
\end{theorem}
 \begin{figure}[h]
 \vspace*{-0.6cm}

\begin{center}
\mbox{ \ }
\centering
\includegraphics[width=8cm]{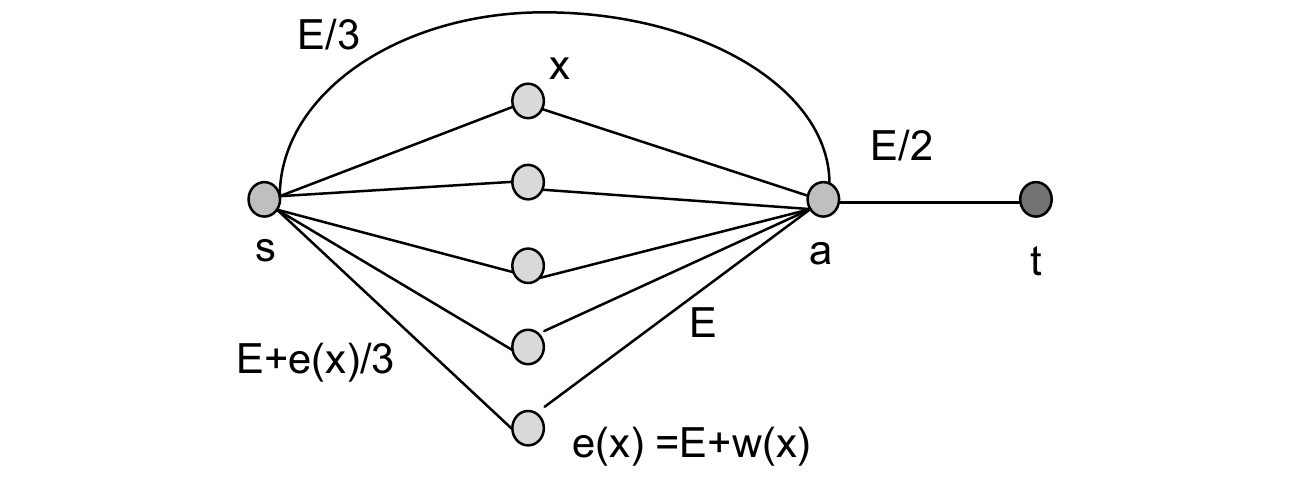}
\caption{
The undirected version of the graph from Figure~\ref{graph2}. 
The weights of middle nodes and lengths of edges incident to middle nodes are
increased by $E$.
}\label{graph1}
\end{center}
\vspace*{-1cm}

\end{figure}

\begin{proof}
We take an undirected version of the graph from the previous proof, see Figure~\ref{graph1}.
Additionally for each $x\in X$ we increase by $E$ the energy of $x$ and edges $x\rightarrow s,
\, x\rightarrow a$.
The agents are placed initially in each node of the graph.

%
\noindent {\bf Delivery.} Due to the drastic increase of the lengths of edges
in the delivery from $s$ to $t$
the edges from $X$ to other nodes can be used only once and in one direction (from $X$
to outside part). 

\noindent {\bf Convergecast.} Now we can take $t$ as the convergast node.
The problem reduces to the delivery from $s$ to $t$, since all energy is in $X$ and
each node contains an agent.
Hence the problem is reduced to the delivery in the directed graph
from the previous point.

\noindent {\bf Broadcast.} The broadacast from $s$ reduces to the delivery from $s$ to $t$,
since in the delivery all agents from $X$ can move to $s$ or to $a$.
The information from $s$, if it has to arrive at $t$, should also arrive at $a$, where
agents coming from $X$ to $a$ might get it.
\end{proof}
\section{Final Remarks}
\vspace*{-0.15cm}
It is rather surprising that, without energy
exchange, even the simplest problem of data delivery is NP-complete
in the simplest  environment of the line, while, as we have shown in this paper, all considered communication problems with energy exchange
are solvable in polynomial time even for tree networks.
On the other hand it is not surprising that  energy exchange in
general graphs does not help and the problems are NP-complete.

\end{document}